\documentclass{endm}
\usepackage{endmmacro}
\usepackage{graphicx}
\usepackage{tikz}

\begin{document}

\begin{verbatim}\end{verbatim}\vspace{2.5cm}

\begin{frontmatter}

\title{Characterization by forbidden induced subgraphs of some subclasses of chordal graphs$^0$}

\author{S\'{e}rgio H. Nogueira\thanksref{ALL}\thanksref{myemail}}
\address{IEF , PPGMMC \\  Universidade Federal de Vi\c cosa, Centro Federal de Educa\c c\~ao Tecnol\'ogica de Minas Gerais\\ Belo Horizonte, Brazil}

\author{ Vinicius F. dos Santos\thanksref{ALL}\thanksref{coemail}}
\address{DCC, PPGMMC\\  Universidade Federal de Minas Gerais,  Centro Federal de Educa\c c\~ao Tecnol\'ogica de Minas Gerais\\
   Belo Horizonte, Brazil} \thanks[ALL]{ Partially supported by CAPES, CNPq and Fapemig} \thanks[myemail]{Email:
   \href{sergiohnog@gmail.com} {\texttt{\normalshape
   sergiohnog@gmail.com}}} \thanks[coemail]{Email:
   \href{viniciussantos@dcc.ufmg.br} {\texttt{\normalshape
   viniciussantos@dcc.ufmg.br}}}

\begin{abstract}
Chordal graphs are the graphs in which every cycle of length at least four has a chord. A set $S$ is a vertex separator for vertices $a$ and $b$ if the removal of $S$ of the graph  separates $a$ and $b$ into distinct  connected components.
A graph \emph{G} is \emph{chordal} if and only if every minimal vertex separator is a clique.
We study subclasses of chordal graphs defined by restrictions imposed on   the intersections of its minimal separator cliques. 
Our goal is to characterize them by forbidden induced subgraphs. Some of these classes have already been studied such as chordal graphs in which two minimal separators have no empty intersection if and only if they are equal.
Those graphs are known as \emph{strictly chordal graphs} and they were first introduced as block duplicate graphs by Golumbic and Peled ~\cite{golumbic}, they  were also considered in ~\cite{kennedy} and ~\cite{pablo2}, showing that strictly chordal graphs are exactly the (gem, dart)-free graphs.\footnote[0]{Accepted for presentation at EURO/ALIO 2018}

\end{abstract}

\begin{keyword}
chordal graphs, forbidden induced subgraphs, minimal vertex separators.
\end{keyword}

\end{frontmatter}

\section{Introduction}\label{intro}
Throughout this text, a graph is always simple, finite and undirected. For sets $R$ and $S$, we denote $R \subseteq S$ if $R$ is a subset of $S$, and $R \subset S$ if $R$ is a proper subset of $S$. For a graph $G$ we denote the set of vertices by $V(G)$ and the set of edges by $E(G)$.
 Given two non-adjacent vertices $u$ and $v$ in the same connected component of $G$, a \emph{uv-separator} is a set $S$ contained in $V(G)$ such that $u$ and $v$ are in different connected components of $G \setminus S$. This separator $S$ is minimal if no proper subset of $S$ is also a \emph{uv-separator}. We will just say \emph{minimal vertex separator} to refer to a set $S$ that is a $uv$-minimal separator for some pair of non-adjacent vertices $u$ and $v$ in $G$.
    A graph is \emph{chordal} if it has no cycle of length at least four as induced graph.
    A \emph{clique} is a maximal set of pairwise adjacent vertices.
    A \emph{clique tree} of a connected chordal graph is any tree $\mathcal{T}$ whose vertices are the  cliques of $G$ such that for every two cliques $C_1,C_2$ each clique on the path from $C_1$ to $C_2$ in $\mathcal{T}$ contains $C_1\cap C_2$.

In this paper we study the relationship between the family of minimal vertex separators and the structure of the graph. In particular we characterize the graph classes arising from properties imposed on the family of minimal vertex separators. 

An important characterization of chordal graphs is due to Dirac:
\begin{theorem}[\cite{dirac}]
A graph $G$ is chordal if and only if every minimal vertex separator of $G$ is a clique.
\end{theorem}
Two  cliques $C_1,C_2$ of G form a  \emph{ separating pair} if $C_1 \cap C_2 $ is non-empty, and every path in G
from a vertex of $C_1 \setminus  C_2$ to a vertex of $C_2\setminus C_1$ contains a vertex of $C_1 \cap C_2$.
Every minimal separator of a chordal graph G is a clique and moreover, it is precisely the intersection of two cliques, as follows.

\begin{theorem}[\cite{habib}]
A set $S$ is a minimal separator of a chordal graph $G$ if and only if there exist maximal cliques $C_1,C_2$ forming a separating pair such that $S=C_1\cap C_2$.
  \end{theorem}
A chordal graph can also be characterized using clique trees as follows:
\begin{theorem}[\cite{gavril}]
A graph is chordal if and only if it has a clique tree.
\end{theorem}
 Let $G$ be a chordal graph and let  $\mathcal{T}(G)$ be a clique tree of $G$. The edges of $\mathcal{T}(G)$ can be labeled
  as the intersection of the endpoints, and these labels are exactly the minimal vertex separators.
We denote by $\mathbf{S}_\mathcal{T}(G)$ the multiset of labels  of the edges of $\mathcal{T}(G)$.
A graph can have several  distinct clique trees and when we are studying chordal graphs, we have the next result   giving  us that the multiset  $\mathbf{S}_\mathcal{T}(G)$  is independent of the clique tree.

\begin{theorem}[\cite{blair}]\label{separadores}
Let $G$ be a chordal graph. The multiset $\mathbf{S}_\mathcal{T}(G)$ of minimal vertex separators of $G$ is the same for every clique tree  $\mathcal{T}(G)$.
\end{theorem}

In the light of Theorem~\ref{separadores} from now on we omit the subscript and use simply~$\mathbf{S}(G)$.

A graph class $\mathcal{G}$ is hereditary if, for every $G \in \mathcal{G}$ and every induced subgraph $H$ of $G$, $H\in \mathcal{G}$.

Our goal is to characterize   hereditary subclasses  of chordal graphs by the intersection and containment relations of their minimal vertex separators and by forbidden induced subgraphs. 

As we shall see, not every restriction on the minimal vertex separators leads to a hereditary graph class. In order to obtain characterization by forbidden induced subgraphs we will impose the additional requirement of being hereditary.

\section{Results}

We start with an auxiliary result showing that minimal vertex separators are, in some sense, hereditary.
\begin{lemma}
\label{separador}
Let $G$ be a chordal graph, $S$ be a minimal vertex separator of $G$, $R$ be a proper subset of $S$ and $S'=S \setminus  R$. Then there exist cliques $C_1', C_2'$ in $ G \setminus R$ such that $S'$ separates $C_1'$ and $C_2'$.
\end{lemma}
\begin{proof}
Since $S$ is a minimal vertex separator of $G$, there exist cliques $C_1, C_2$ in $G$ such that $S=C_1\cap C_2$. Let $C_1', C_2'$  be cliques in $G'= G \setminus R$ such that $(C_i \setminus R )\subset C_i'$, $i=1,2$.
 $S'$ separates $C_1'$ and $C_2'$, because if there exists a path in $G'$ between a vertex $u_1 \in C_1'$ and $u_2 \in C_2'$ in $G'$, then there exists a path from a vertex of $C_1$ to a vertex of $C_2$ in $G \setminus S$, contradicting the fact that $S$ separates $C_1$  and $C_2$ in $G$. Now suppose that $S'$ is not minimal and let $R'$ be a non-empty subset of $S'$ such that $S'\setminus R'$ separates $C_1',C_2'$. Then $S\setminus R'$ separates  $C_1,C_2$, contradicting the fact that $S$ is minimal vertex  separator of $G$.
\end{proof}

In the  following, we will consider the graphs depicted in Figure 1:

\begin{center}
\begin{tikzpicture}
\draw (0,0) -- (1,0);
\draw (0,0) -- (1,0.5);
\draw(0,0) -- (1,-0.5);
\draw [fill] (0,0) circle [radius=0.05];
\draw[fill](1,0.5) circle [radius=0.05];
\draw[fill] (1,0) circle [radius=0.05];
\draw[fill] (1,-0.5) circle [radius=0.05];
\node[below] at (0.5,-1) {claw};
\draw (2,0) -- (3,0.25);
\draw (2,0) -- (3,-0.25);
\draw (2,0) -- (2.75,0.75);
\draw (2,0) -- (2.75,-0.75);
\draw [fill] (2,0) circle [radius=0.05];
\draw[fill](3,0.25) circle [radius=0.05];
\draw[fill] (2.75,0.75) circle [radius=0.05];
\draw[fill] (2.75,-0.75) circle [radius=0.05];
\draw[fill](3,-0.25) circle [radius=0.05];
\draw (2.75,0.75) -- (3,0.25);
\draw (3,0.25) -- (3,-0.25);
\draw (3,-0.25) -- (2.75,-0.75);
\node [below] at (2.5,-1) {gem};
\draw [fill] (4.5,-1) circle [radius=0.05];
\draw[fill](4.5,0.15) circle [radius=0.05];
\draw[fill] (4.5,1) circle [radius=0.05];
\draw[fill] (3.9,-0.15) circle [radius=0.05];
\draw[fill](5.1,-0.15) circle [radius=0.05];
\draw (4.5,-1) -- (4.5,0.15);
\draw (4.5,0.15) -- (4.5,1);
\draw (4.5,1) -- (3.9,-0.15);
\draw (4.5,1) -- (5.1,-0.15);
\draw (3.9,-0.15) -- (4.5,0.15);
\draw (4.5,0.15) -- (5.1,-0.15);
\node[below] at (4.5,-1) {dart};
\draw (6,-0.865) -- (7,-0.865);
\draw (6,-0.865) -- (6.5,0);
\draw (7,-0.865) -- (6.5,0);
\draw (7.5,0) -- (6.5,0);
\draw (7.5,0) -- (7,-0.865);
\draw (7.5,0) -- (8,-0.865);
\draw [fill] (6,-0.865) circle [radius=0.05];
\draw[fill](6.5,0) circle [radius=0.05];
\draw[fill] (7,-0.865) circle [radius=0.05];
\draw[fill]  (7,-0.865) circle [radius=0.05];
\draw[fill](7.5,0) circle [radius=0.05];
\draw[fill] (8,-0.865) circle [radius=0.05];
\draw[fill] (7,-0.865) -- (8,-0.865);
\draw[fill](7,0.865) circle [radius=0.05];
\draw (7,0.865) -- (6.5,0);
\draw (7,0.865) -- (7.5,0);
\node[below] at (7,-1) {Haj\'{o}s};
\draw (9,0) -- (10,0);
\draw(10,0) -- (11,0);
\draw (9,0.75) -- (10,0);
\draw (9,-0.75) -- (10,0);
\draw (10,0) -- (11,0);
\draw (10,0) -- (11,0.75);
\draw (10,0) -- (11,-0.75);
\draw (9,0.75) -- (9,-0.75);
\draw (11,0.75) -- (11,-0.75);
\draw [fill] (9,0) circle [radius=0.05];
\draw[fill](10,0) circle [radius=0.05];
\draw[fill] (11,0) circle [radius=0.05];
\draw[fill]  (9,0.75) circle [radius=0.05];
\draw[fill](11,0.75) circle [radius=0.05];
\draw[fill](9,-0.75) circle [radius=0.05];;
\draw[fill] (11,-0.75) circle [radius=0.05];
\node[below] at (10,-1) {butterfly};

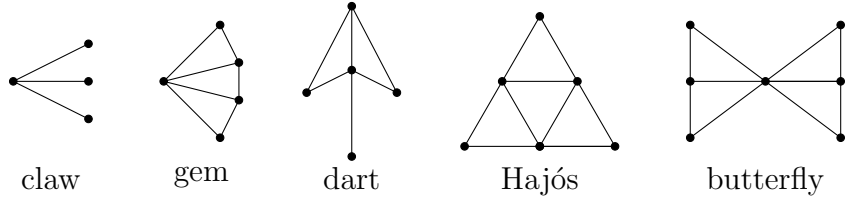
\node at (5.5,-1.9) {Figure 1 - Forbidden Induced subgraphs considered in this work};
\end{tikzpicture}

\end{center}

\begin{lemma}
\label{disjuntos}
Let $G$ be a chordal graph. The following statements are equivalent:
\begin{itemize}
\item
i) For every   induced subgraph $ G'$  of $G$   and for every pair $S_i, S_j \in \mathbf{S}(G')$,  $S_i\cap S_j \neq \emptyset$;
\item
ii) $G$ is  $(P_4,2P_3)$-free.
\end{itemize}
\end{lemma}

\begin{proof}
 Let $G$ be a graph satisfying $i)$   and suppose that $ii)$ is not valid.
Suppose that there exists an induced subgraph $G'$ of $G$ isomorphic to $P_4$ with vertices $v_1, v_2, v_3,v_4$ and cliques $C_1=\{v_1,v_2\}$, $C_2=\{v_2,v_3\}$ and $C_3=\{v_3, v_4\}$. Then we have $S_1= C_1 \cap C_2=\{v_2\}$, $S_2=C_2 \cap C_3=\{v_3\}$ and  $S_1\cap S_2 = \emptyset$, a contradiction.
Now suppose we have an induced subgraph $G'$ of $G$ isomorphic to $2P_3$ with vertices $v_1v_2 v_3$ and $v_4v_5v_6$  and cliques $C_1=\{v_1,v_2\}$, $C_2=\{v_2,v_3\}$, $C_3=\{v_4, v_5\}$ and $C_4=\{v_5, v_6\}$. Then we have $S_1= C_1 \cap C_2=\{v_2\}$, $S_2=C_3 \cap C_4=\{v_5\}$ and  $S_1\cap S_2 = \emptyset$, a contradiction. Hence $G$ is $(P_4,2P_3)$-free.
Therefore $i) \Rightarrow ii)$.

Conversely, let $G$ be a graph satisfying $ii)$   and suppose that $i)$ is not valid.
Let $ G'$ be an  induced subgraph of $G$. Let $\mathcal{T'}$ be a clique tree of $G'$  and $\mathbf{S}(G')$ be the multiset of minimal vertex separators of $G'$. Suppose that there  exist $S_1, S_2 \in \mathbf{S}(G')$, with $S_1\cap S_2 = \emptyset$.
First, suppose that there  exist adjacent edges $e_1, e_2 \in E(\mathcal{T'})$ with labels $S_1,S_2$  and let $C_1,C_2,C_3$ be cliques such that $S_1=C_1\cap C_2$ and $ S_2=C_2\cap C_3$. Suppose $x\in S_1$ and $y\in S_2$. Since the cliques are maximal there must exist $a \in $ $C_1 \setminus C_2$ and $b \in $ $C_3 \setminus C_2$. Then $\{a,x,y,b\}$ induces a $P_4$.
Now suppose that there exist non-adjacent edges $e_1, e_2 \in E(\mathcal{T'})$  with labels $S_1,S_2$ and let $C_1,C_2,C_3, C_4$ be cliques such that $S_1=C_1\cap C_2$ and $ S_2=C_3\cap C_4$. Without loss of generality we can consider that the  path in $\mathcal{T'}$ from $\{C_1,C_2\}$ to $\{C_3,C_4\}$ contains $C_2$ and $C_3$. Suppose $x\in S_1$ and $y\in S_2$. Since the cliques are maximal we must have $a \in $ $C_1 \setminus C_2$, $b \in $ $C_4 \setminus C_3$ and  $c \in $ $C_2 \setminus C_1$. If $xy \in E(G')$ then $\{a,x,y,b\}$ is an  induced  $P_4$ of $G$; otherwise if $cy \in E(G')$ then $\{a,x,c,y\}$  induces a $P_4$; else  there exists $d \in C_3 \setminus \{a,b,c,x,y\}$. If $dx\in E(G') $ then $\{x,d,y,b\}$  is $P_4$; if $cd \in E(G')$ then $\{x,c,d,y\}$  is $P_4$  else $\{a,x,c\}$, $\{d,y,b\}$ induce $2P_3$, a contradiction. Therefore $ii) \Rightarrow i)$.
\end{proof}

\begin{lemma}
\label{iguais}
Let $G$  be a chordal graph. The following statements are equivalent:
\begin{itemize}
\item
i) For every $ G'$  induced subgraph of $G$   and  for every pair $S_i, S_j \in \mathbf{S}(G')$,  $S_i\neq S_j$.
\item
ii) $G$ is claw-free.
\end{itemize}
\end{lemma}

\begin{proof}
Let $G$ be a graph satisfying $i)$   and suppose that $ii)$ is not valid.
Let $G'$  be an induced subgraph of $G$ and let $\{x\}\{a,b,c\}$ be an induced claw of $G'$, with cliques $C_1=\{x,a\}, C_2=\{x,b\}, C_3=\{x,c\}$. Then we have $S_1=C_1 \cap C_2= \{x\}=C_2 \cap C_3=S_2$, a contradiction. Hence $G$ is claw-free and $i) \Rightarrow ii)$.

Conversely let $G$ be a graph satisfying $ii)$   and suppose that $i)$ is not valid. Let $G'$ be an induced subgraph of $G$.
Let  $\mathcal{T'}$ be a clique tree of $G'$, $\mathbf{S}(G')$ be the multiset of minimal vertex separators of $G'$ and suppose that  $\exists S_1, S_2 \in \mathbf{S}(G')$, with $S_1=S_2$.
First, suppose that there exist adjacent edges $e_1, e_2 \in E(\mathcal{T'})$ with labels $S_1,S_2$ and let $C_1,C_2,C_3$ be cliques such that $S_1=C_1\cap C_2$ and $ S_2=C_2\cap C_3$. Since the separators are equal, take $x\in S_1 \cap S_2$. Since the cliques are maximal there must exist $a \in $ $C_1 \setminus C_2$, $b \in $ $C_3 \setminus C_2$ and $c \in C_2 \setminus (C_1 \cup C_3)$. Then $\{x,a,b,c\}$ induces a claw, a  contradiction.
Now suppose that  $\nexists$ $e_1, e_2 \in E(\mathcal{T'})$ adjacent with labels $S_1,S_2$ and let $C_1,C_2,C_3, C_4$ be distinct cliques such that $S_1=C_1\cap C_2$ and $ S_2=C_3\cap C_4$.  Without loss of generality we can consider that every path in $\mathcal{T'}$ from $\{C_1,C_2\}$ to $\{C_3,C_4\}$ contains $C_2$ and $C_3$. Let $x\in \bigcap _{i=1,...,4} C_i$. Since the cliques are maximal we must have $a \in $ $C_1 \setminus C_2$, $b \in C_4 \setminus C_3$ and  $c \in $ $C_2 \setminus C_1$ and  then      $\{x,a,b,c\}$ induces a claw, a  contradiction. Hence   $ii) \Rightarrow i)$.
\end{proof}

Due to space restrictions the proofs of the following two lemmas will be omitted.

\begin{lemma}
\label{cont.estrito}
Let $G$  be a  chordal graph. The following statements are equivalent:
\begin{itemize}
\item
i) For every $ G'$  induced subgraph of $G$   and for every pair $S_i, S_j \in \mathbf{S}(G')$,  $S_i$ is not strictly contained in $S_j$;
\item
ii) $G$ is dart-free.
\end{itemize}
\end{lemma}



\begin{lemma}
\label{n.contido}
Let $G$  be a chordal graph. The following statements are equivalent:
\begin{itemize}
\item
i) For every $ G'$  induced subgraph of $G$   and for every pair $S_i, S_j \in \mathbf{S}(G')$, $S_i \subseteq S_j$ or  $S_j \subseteq S_i$;
\item
ii) $G$ is (gem, butterfly)-free.
\end{itemize}
\end{lemma}

Let $\mathcal{H}$  be a hereditary subclass of chordal graphs. Let $G\in \mathcal{H}$, $\mathcal{T}$ be a clique tree of $G$  and $\mathbf{S}(G)$ be the multiset of minimal vertex separators of $G$. For each pair $S_i,S_j \in \mathbf{S}(G)$ one of the following situations  holds:
\begin{itemize}
\item
(a) \textbf{Disjunction:}  $S_i\cap S_j=\emptyset$.
\item
(b) \textbf{Equality:}   $S_i=S_j$.
\item
(c) \textbf{Containment:}  $S_i \subset S_j$ or $S_j \subset S_i$.
\item
(d) \textbf{Overlap:}  $S_i\nsubseteq S_j$ and $S_j\nsubseteq S_i$.
\end{itemize}

\begin{remark}
  Since we are interested in a hereditary class  $\mathcal{H}$, we can note that,  by Lemma \ref{separador}, if we have Containment then we must allow  Equality and  if we have  Overlap then we must allow  Equality and  Disjunction.
\end{remark}

\begin{remark}
\label{free}
And again by  hereditarity we can note that if a class is claw-free then it is dart-free;  if it is $P_4$-free  then it is gem-free and if it is dart-free or $2P_3$ -free  then it is butterfly-free.
\end{remark}

Hence all possible combinations of properties $(a)-(d)$ are characterized in the following theorem.
\begin{theorem}
Let $G$  be a  chordal graph. Then for every $ G'$  induced subgraph of $G$   and for  every pair $S_i, S_j \in \mathbf{S}(G'),  i\neq j$, we have:
\begin{itemize}
\item (i)   $S_i\cap S_j=\emptyset \Leftrightarrow$ $G$ is (claw, gem)-free.
\item (ii)   $S_i=S_j$ $\Leftrightarrow$ $G$ is ($P_4$, gem, butterfly)-free.
\item (iii)  $S_i\cap S_j=\emptyset$ or $S_i=S_j$ $\Leftrightarrow$ $G$ is (dart, gem)-free.
\item (iv)   $S_i\cap S_j=\emptyset$ or $S_i=S_j$ or $S_i \subset S_j$ or $S_j \subset S_i \Leftrightarrow$ $G$ is (gem, butterfly)-free.
\item (v)   $S_i\cap S_j=\emptyset$ or $S_i=S_j$ or $(S_i\nsubseteq S_j$ and $S_j\nsubseteq S_i)$ $\Leftrightarrow$ $G$ is dart-free.
\item (vi)    $S_i=S_j$ or $S_i\subset S_j$ or $S_j\subset S_i\Leftrightarrow$ $G$ is $(2P_3, P_4)$-free.
 \end{itemize}
 \end{theorem}

\begin{proof}
 It follows from Lemmas \ref{disjuntos}-\ref{n.contido} and Remark \ref{free}.
\end{proof}

We remark that $(iii)$ had been previously proved in ~\cite{kennedy} and ~\cite{pablo2}.

We now move our attention to Helly Property. Let $\mathcal{F}$  be a family of subsets of a set $\mathcal{S}$. We say that $\mathcal{F}$ satisfies the Helly property when
every subfamily $\mathcal{F}'$ of $\mathcal{F}$ consisting of pairwise intersecting subsets satisfies
$\cap_{F\in\mathcal{F}'} F \not= \emptyset$.

We say that a set $X$ is a witness that the Helly property does not hold if:
\begin{itemize}
\item
$S_i \cap S_j \neq \emptyset$, $\forall S_i,S_j \in X$ and
\item
$\bigcap_{S_i \in  X}S_i = \emptyset$.
\end{itemize}

\begin{lemma}
Let $G$ be a chordal graph such that $\mathbf{S}(G)$ does not satisfy the Helly property and such that for every induced subgraph $G'$ the family $\mathbf{S}(G')$ satisfies the Helly property. 
Let $\mathcal{T}$ be a clique tree of $G$ and $\mathbf{S_{\ell}}(T)$ be the set of minimal vertex separators incident to leaves of $\mathcal{T}$. Then $\mathbf{S_{\ell}}(T)$ is a  witness.
\end{lemma}

\begin{proof}
If there exist separators $S_i,S_j \in \mathbf{S_{\ell}}(T)$ with $S_i \cap S_j = \emptyset$ then $S_i,S_j$ cannot be simultaneously in a witness. But then there exists witness $R \subset \mathbf{S}(G) \backslash S_i$ or $R \subset \mathbf{S}(G) \backslash S_j$, wich contradicts the minimality of $G$. Then suppose that  $S_i \cap S_j \neq\emptyset$ for all pairs $S_i,S_j \in \mathbf{S_{\ell}}(T)$. Since $\mathbf{S_{\ell}}(T)$ is not a witness there exists $x \in \bigcap_{S_i \in  \mathbf{S_{\ell}}(T)}S_i$ and then $x$ is universal in $G$ and this implies that $x$ belongs to all minimal vertex  separators of $G$. But this implies that $\mathbf{S}(G)$ satisfies the Helly property, a contradiction. 
\end{proof}

\begin{theorem}
Let $\mathcal{G}$ be the hereditary class of chordal graphs such that $\forall G \in \mathcal{G}$,   $\mathbf{S}(G)$ satisfies the Helly property. Then for any chordal graph $G$,
 $ G \in \mathcal{G} \Leftrightarrow$ G is Haj\'{o}s-free.
\end{theorem}

\begin{proof}
Let $G$ $\in  \mathcal{G}$ and let $xyzabc$  be an induced Haj\'{o}s of $G$, with cliques $C_1=\{x,y,a\}, C_2=\{x,z,b\}, C_3=\{y,z,c\}, C_4=\{x,y,z\}$. Let $S_1=C_1\cap C_4=\{x,y\}, S_2=C_2\cap C_4=\{x,z\}, S_3=C_3\cap C_4=\{y,z\}$. Then we have $S_1 \cap S_2=\{x\}, S_1 \cap S_3=\{y\}, S_2 \cap S_3=\{z\}$ and $S_1\cap S_2\cap S_3=\emptyset $, so $\mathbf{S}(G)$ does not satisfy the Helly property.

Conversely let  $G \not\in \mathcal{G}$ and let $G' $ be an induced  subgraph of $G$ minimal in relation to the property that $\mathbf{S}(G')$ does not satisfy the Helly property  and let  $\mathcal{T'}$ be a clique tree of $G'$. By the previous Lemma,  $\mathbf{S_{\ell}}(T')$ is a witness.
Let $S_1,S_2 \in \mathbf{S_{\ell}}(T')$.  Then $\exists x_1 \in S_1\cap S_2$ and $\exists S_i \in \mathbf{S_{\ell}}(T')$ such that $x_1 \notin S_i$. 
Without lost of generality, let $S_i=S_3$. Note that if $A,B,C \in \mathbf{S_{\ell}}(T')$ then $A\cap B \not\subseteq  C$. Indeed, suppose $A\cap B \subseteq  C$. Since $\mathbf{S_{\ell}}(T')$ is a witness, we know that
$\bigcap_{X \in  \mathbf{S_{\ell}}(T')}X = \emptyset$. Then we have: 
$$\left (\bigcap_{X \in  \mathbf{S_{\ell}}(T') \setminus \{A,B,C\}}X \right ) \cap A \cap B \cap C = \emptyset \Rightarrow $$
$$\left (\bigcap_{X \in  \mathbf{S_{\ell}}(T') \setminus \{A,B,C\}}X  \right) \cap A \cap B  = \emptyset \Rightarrow $$
$$\bigcap_{X \in  \mathbf{S_{\ell}}(T') \setminus \{C\}}X = \emptyset.$$
But this implies that there exists a proper subset of $\mathbf{S_{\ell}}(T')$ that does not satisfy the Helly property, contradicting the minimality of $G'$. Then we can suppose that $\exists y \in S_3 \cap S_1 \backslash S_2$ and $\exists z\in S_3 \cap S_2 \setminus S_1$. Let $C_1, C_2, C_3$ be the leaves of edges labeled by $S_1,S_2,S_3$ respectively in $\mathcal{T'}$. Let $v_i$ be an exclusive vertex of $C_i, i=1...3$. Then $\{x,y,z, v_1,v_2,v_3\}$  induces a Haj\'{o}s.
\end{proof}

\end{document}